\documentclass[11pt]{article}    

\usepackage[margin=1in]{geometry}
\usepackage[utf8]{inputenc} 
\usepackage{amsfonts, amsmath, amsthm, amssymb}       
\usepackage{mlmodern}
\usepackage{nicefrac}       
\usepackage{microtype}      
\usepackage[shortlabels]{enumitem}
\usepackage{fontawesome}
\usepackage{xcolor} 
\usepackage{hyperref}     
\usepackage{url}          
\usepackage{graphicx}

\usepackage{etoolbox}
\newtoggle{extended}
\toggletrue{extended}

\definecolor{otherlightblue}{RGB}{0, 100, 200}
\definecolor{othergreen}{RGB}{30, 125, 0}
\definecolor{otherblue}{RGB}{0, 60, 150}

\hypersetup{
	colorlinks=true,
	pdfpagemode=UseNone,
	citecolor=othergreen,
	linkcolor=otherblue, 
	urlcolor=magenta
}

\usepackage[font=small]{caption}
\usepackage[capitalize, nameinlink]{cleveref}

\newcommand{\lipschitzcost}{L_V}
\newcommand{\phitilde}{\widetilde{\phi}}

\newtheorem{theorem}{Theorem}[section] 
\newtheorem{definition}[theorem]{Definition}
\newtheorem{lemma}[theorem]{Lemma}
\newtheorem{fact}[theorem]{Fact} 
\newtheorem{corollary}[theorem]{Corollary}

\newtheorem{assumption}[theorem]{Assumption}

\newcommand{\nconstr}{m}
\newcommand{\T}{T}

\newcommand{\R}{\mathbb{R}}

\newcommand{\Diag}{\mathrm{Diag}}

\newcommand{\fcl}[1]{f_{\mathrm{cl}}^{#1}}
\newcommand{\E}{\mathbb{E}}

\newcommand{\pirs}{\pi^{rs}}
\newcommand{\piexpert}{\pi^{\star}} 
\newcommand{\pilearned}{\hat{\pi}}

\DeclareMathOperator*{\argmin}{arg\,min}
\DeclareMathOperator*{\adj}{adj}

\newcommand{\X}{{X}}
\newcommand{\U}{{U}}

\newcommand{\pimpc}{\boldsymbol{\pi}_{\mathrm{mpc}}}
\newcommand{\pibmpc}{\boldsymbol{\pi}_{\mathrm{mpc}}^\eta}

\newcommand{\vueta}{\vec{u}_{\eta}}

\newcommand\numberthis{\addtocounter{equation}{1}\tag{\theequation}} 
\numberwithin{equation}{section}

\newtheorem{problem}[theorem]{Problem}

\crefname{prob}{Problem}{Problems}
\crefname{ineq}{Inequality}{Inequalities}
\crefformat{equation}{(#2#1#3)}

\newcommand{\mat}[1]{{#1}}
\renewcommand{\vec}[1]{{#1}}

\newcommand{\mI}{\mat{I}}

\newcommand{\mLambda}{\mat{\Lambda}}
\newcommand{\mA}{\mat{A}}

\newcommand{\mD}{\mat{D}}

\newcommand{\mH}{\mat{H}}
\newcommand{\mF}{\mat{F}}

\newcommand{\mG}{\mat{G}}
\newcommand{\mP}{\mat{P}}
\newcommand{\vw}{\vec{w}}

\newcommand{\mM}{\mat{M}}

\newcommand{\DuetaDx}{\frac{\partial \vueta}{\partial \vec{x}_0}}

\newcommand{\va}{\vec{a}}
\newcommand{\vd}{\vec{d}}
\newcommand{\vb}{\vec{b}}
\newcommand{\vx}{\vec{x}}
\newcommand{\vu}{\vec{u}}
\newcommand{\vv}{\vec{v}}
\newcommand{\ve}{\vec{e}}

\newcommand{\vg}{\vec{g}}
\newcommand{\mK}{\mat{K}}
\newcommand{\mC}{\mat{C}}
\newcommand{\mU}{\mat{U}}
\newcommand{\mV}{\mat{V}}

\renewcommand{\det}{\mathrm{det}}

\newcommand{\philb}{\textrm{res}_{\textrm{l.b.}}}
\global\long\def\vc{\vec{c}}%
\global\long\def\vx{\vec{x}}%
\global\long\def\gap{\mathrm{gap}}%
\global\long\def\u{\vec{u}}%
\global\long\def\ue{\u_{\eta}}%
\global\long\def\v{\vec{v}}%
\global\long\def\vstar{\v^{\star}}%
\global\long\def\veta{\v_{\eta}}%
\global\long\def\us{\u^{\star}}%
\global\long\def\re{\textrm{res}}%

\title{
On the Sample Complexity of Imitation Learning\\ for Smoothed Model Predictive Control\thanks{The first two authors contributed equally. A  preliminary version of this manuscript is published in CDC 2024.}
}

\author{Daniel Pfrommer\thanks{Massachusetts Institute of Technology. Email: \texttt{dpfrom@mit.edu}. }
 \and Swati Padmanabhan\thanks{Massachusetts Institute of Technology. Email: \texttt{pswt@mit.edu}.}
 \and Kwangjun Ahn\thanks{Microsoft Research. Email: \texttt{kwangjunahn@microsoft.com}.}
 \and Jack Umenberger\thanks{University of Oxford. Email: \texttt{jack.umenberger@eng.ox.ac.uk}.}
 \and Tobia Marcucci\thanks{Massachusetts Institute of Technology. Email: \texttt{tobiam@mit.edu}.}
 \and Zakaria Mhammedi\thanks{Massachusetts Institute of Technology. Email: \texttt{mhammedi@mit.edu}.}
 \and Ali Jadbabaie\thanks{Massachusetts Insitute of Technology. Email: \texttt{jadbabai@mit.edu}.}
 }

\usepackage[
backref=true,
natbib=true,
style=numeric,
giveninits=false, 
maxbibnames = 99, 
maxalphanames=8,
sorting=none
]{biblatex}

\begin{document}

\maketitle

\begin{abstract}

Recent work in imitation learning has shown that having an expert controller that is both suitably smooth and stable enables stronger guarantees on the performance of the  learned controller. However, constructing such smoothed expert controllers for arbitrary systems remains challenging, especially in the presence of input and state constraints. As our primary contribution, we show how such a smoothed expert can be designed for a general class of systems using a log-barrier-based relaxation of a standard Model Predictive Control (MPC) optimization problem. At the crux of this theoretical guarantee on  smoothness is a new lower bound we prove on the optimality gap of the analytic center associated with a convex Lipschitz function, which we hope could be of independent interest. We validate our theoretical findings via experiments, demonstrating the merits of our smoothing approach over randomized smoothing.

\end{abstract}

\newpage

\section{Introduction}\label{sec:introduction}
Imitation learning has emerged as a powerful tool in machine learning, enabling agents to learn complex behaviors  by imitating expert demonstrations  acquired either from a human demonstrator or a policy computed offline~\cite{pomerleau1988alvinn, ratliff2009learning, abbeel2010autonomous, ross2011reduction}. Despite its significant success, imitation learning often suffers from a {compounding error problem}: Successive evaluations of the approximate policy can accumulate error, resulting in out-of-distribution failures
\cite{pomerleau1988alvinn}. Recent results~\cite{pfrommer2022tasil, tu2022sample, block2023provable} have identified \emph{smoothness} (i.e. the derivative, with respect to the state, of the control policy being Lipschitz) and \emph{stability} of the expert as two key properties that enable circumventing this issue, thereby allowing for end-to-end performance guarantees for the final learned controller.

In this paper, our focus is on enabling such guarantees when the expert being imitated is a Model Predictive Controller (MPC), a powerful class of control algorithms based on solving an optimization problem over a receding prediction horizon~\cite{allgower2012nonlinear}.
In some cases, the solution to this multiparametric optimization problem, known as the explicit MPC representation \cite{bemporad2002explicit},
can be pre-computed. For our setup --- linear systems with polytopic constraints --- the optimal control input is known to be a piecewise affine function of the state.
However, the number of these pieces may grow exponentially with the time horizon and the state and input dimension, which could render pre-computing and storing such a representation  impractical in high dimensions.

While the approximation of a linear MPC controller has garnered significant attention~\cite{maddalena2020neural, chen2018approximating, ahn2023model},  prior works typically approximate the (non-smooth) explicit MPC with a neural network and introduce schemes to enforce the stability of the learned policy. In contrast,  we  construct a smoothed version of the expert and  apply stronger theoretical results for the imitation of a smoothed expert.

Specifically, we demonstrate --- both theoretically and empirically ---  that a log-barrier  formulation of the underlying MPC optimization yields the same desired smoothness properties as its randomized-smoothing-based counterpart, while being  faster to compute. Our barrier MPC formulation replaces the constraints in the MPC optimization problem with ``soft constraints'' using the log-barrier (cf. \Cref{sec:barrier_mpc_all}). 
We show that, when used in conjunction with a black-box imitation learning algorithm, this enables end-to-end guarantees on the performance of the learned policy.

\section{Problem Setup and Background}
\label{sec:notation}
\looseness=-1We first state our notation and setup. The notation $\|{}\cdot{}\|$ refers to the $\ell_2$ norm $\|{}\cdot{}\|_2$. Unless transposed, all vectors are column vectors. 
For a vector $\vx$, we use $\Diag(\vx)$ for the diagonal matrix with the entries of $\vx$ along its diagonal. We use $[n]$ for the set $\{1, 2, \dots, n\}$. Given $M \in \R^{n \times n}$ and $\sigma \in \{0,1\}^n$, we denote by $[M]_{\sigma}$ the principal submatrix, of $M$, with rows and columns $i$ for which $\sigma_{i}=1$. We additionally use $M_{\sigma}^{-1}$ and $\adj(M)_\sigma$ to denote, respectively, the inverse and adjugate (the transpose of the cofactor matrix) of $[M]_\sigma$, appropriately padded with zeros back to the size of $M$, at same location.

\looseness=-1We consider constrained discrete-time linear dynamical systems of the form
\begin{align}
x_{t+1} = Ax_t + Bu_t, \quad x_t \in \X, u_t \in \U, \label{eq:dynamics}
\end{align}
with state $x_t \in \R^{d_x}$ and control-input $u_t \in \R^{d_u}$ indexed by time step $t$, and state and input maps $A \in \R^{d_x \times d_x}$ and $B \in \R^{d_x \times d_u}$. The sets  $\X$ and $\U$, respectively, 
are the compact convex state and input constraint sets given by the polytopes $$\X := \{x \in \R^{d_x}\,\mid\, A_x x \leq b_x\}, \quad \U := \{u\in \R^{d_u} \,\mid\, A_u u \leq b_u \},$$ where $A_x \in \R^{k_x \times d_x}$, $A_u \in \R^{k_u \times d_x}$, $b_x \in \R^{k_x}$, and  $b_u \in \R^{k_u}$. A constraint $f(x)\leq 0$ is said to be ``active'' at $y$ if $f(y)=0$. 
For notational convenience, we overload $\phi$ to compactly denote the vector of constraint residuals for a state $x$ and input $u$ as well as for the sequences $x_{1:T}$ and $u_{0:\T-1}$: 
\[
\phi(x,u) := \begin{bmatrix}b_x - A_x x \\ b_u - A_u u\end{bmatrix}, 
\phi(x_0,u_{0:\T-1}) := \begin{bmatrix}\phi(x_1,u_0) \\ \vdots \\ \phi(x_T, u_{\T-1})\end{bmatrix}.\numberthis\label{eq:phi}
\] 

We consider deterministic state-feedback control policies of the form $\pi: \X \to \U$ and denote the closed-loop system under $\pi$ by
$\fcl{\pi}(x) := Ax + B\pi(x)$. 
We  use $\piexpert$ to refer to the expert policy and $\pilearned$ for its learned approximation. 
\looseness=-1In particular, our
 choice of $\piexpert$ 
in this paper is an MPC with quadratic cost and linear constraints.  The MPC policy is obtained by solving the following minimization problem over future actions $\vu:= \vu_{0:\T-1}$
with quadratic cost in $\vu$ and states $\vx := \vx_{1:\T}$: 
\[ 
\begin{array}{ll}
\vspace{0.5em}
\mbox{minimize}_{\vu} & V(\vx_0, \vu):= \sum_{t=1}^\T x_t^\top Q_t x_t + \sum_{t=0}^{\T-1} u_t^\top R_t u_t \\
\mbox{where} &x_{t + 1} := Ax_t + Bu_t,  \\
&x_T \in \X, u_0 \in \U, \\
&x_t \in \X, u_t \in \U, \; \forall t \in [T-1],
\end{array}\label[prob]{eq:V}\numberthis
\] 
where $Q_t$ and $R_{t-1}$ are positive definite for all $t \in [\T]$.
For a given state $x$, the corresponding input $\pimpc$  of the MPC  is: 
\begin{align}\label[prob]{eq:pi_mpc}
\pimpc(x) := \argmin_{u_0} \min_{u_{1:\T-1}} V(x,u_{0:\T-1}),
\end{align}
where the minimization is  over  the feasible set defined in \Cref{eq:V}.
For $\pimpc$ to be well-defined, we assume that $V(x_0,u)$ has a unique global minimum in $u$ for all feasible $x_0$.

\subsection{Explicit Solution to MPC}\label{sec:explicitMPC}
\begin{figure}
    \centering
    \includegraphics[width=0.5\linewidth]{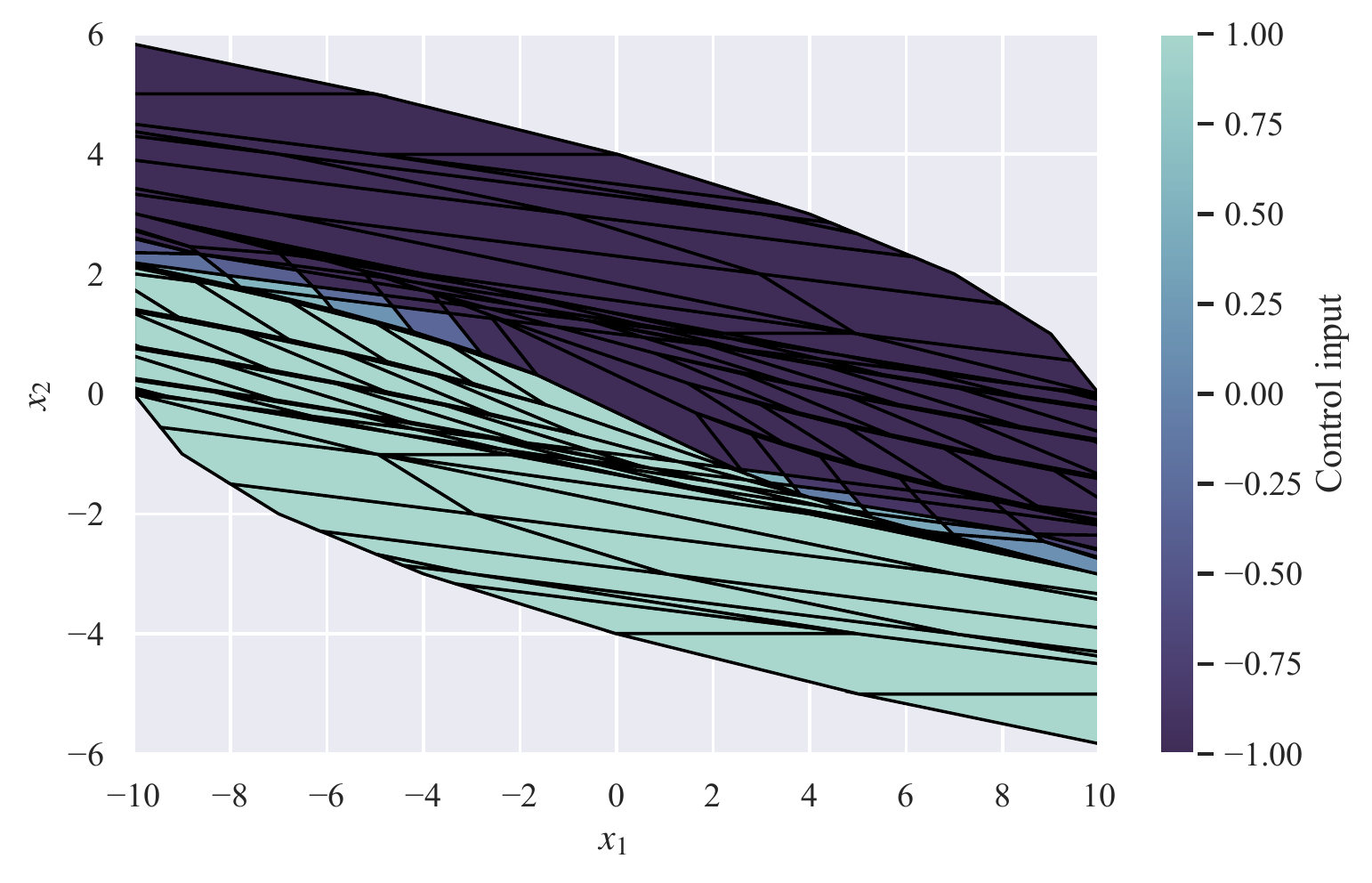}
    \caption{The explicit MPC controller for $A = \begin{bmatrix}1 & 1 \\ 0 & 1\end{bmatrix}, B = \begin{bmatrix}0 \\ 1\end{bmatrix}, Q = I, R = 0.01, \T=10$ with the constraints $\|x\|_\infty \leq 10, |u|\leq 1$.}
    \label{fig:explicit_mpc}
\end{figure}
Explicit MPC \cite{bemporad2002explicit} rewrites  \Cref{eq:pi_mpc} as a multi-parametric quadratic program with linear inequality constraints and  solves it for every possible combination of active constraints, building an analytical solution to the control problem.
We therefore rewrite \Cref{eq:pi_mpc} as the optimization problem, in variable $\vu := \vu_{0:\T-1} \in \R^{\T  d_u}$, as described below:  
\[ 
\begin{array}{ll}
\vspace{0.5em}
\mbox{minimize}_{\vu} &V(x_0, u):= \tfrac{1}{2}\vu^\top \mH \vu  - \vx_0^\top \mF \vu\\
\mbox{subject to } &\mG \vu \leq \vw + \mP \vx_0,  
\end{array}\numberthis\label[prob]{eq:reformulated}
\] 
with  matrices $\mH\in\R^{\T \cdot d_u \times T \cdot d_u}$, $\mF \in \R^{d_x \times \T \cdot d_u}$,  $\mG \in \R^{\nconstr \times \T \cdot d_u}$, and $\mP \in \R^{m \times d_x}$, and vector $\vw\in  \R^{\nconstr}$, all given by 
\begin{align*}
    H &= R_{0:\T-1} + \widehat{B}^\top Q_{1:\T} \widehat{B}, \,\, F = -2\widehat{A}^\top Q_{1:\T}\widehat{B}, \\
    G &= \begin{bmatrix}
        A_u \\
        A_x\widehat{B}
    \end{bmatrix}, P = \begin{bmatrix}
        0 \\
        -A_x\widehat{A}
    \end{bmatrix}, w = \begin{bmatrix}
        b_u \\
        b_x
    \end{bmatrix},
\end{align*}
where $Q_{1:\T}, R_{0:\T-1}$ are block diagonal  with $Q_1, \dots, Q_\T$ and $R_0, \dots, R_{\T-1}$ on the diagonal, and $\widehat{B}$ and $\widehat{A}$ are
\begin{align*}
    \widehat{A} = \begin{bmatrix}A 
    \\ \vdots \\ A^{\T}\end{bmatrix}, \quad
    \widehat{B} = \begin{bmatrix}
        B & 0 & \dots & 0 & \\
        \vdots & \vdots & \ddots & \vdots \\
        A^{\T -1}B & A^{\T-2}B & \dots & B
    \end{bmatrix}
\end{align*}
so that $x_{1:\T} = \widehat{A}x_0 + \widehat{B}u$. 
We assume that the constraint polytope in \Cref{eq:reformulated} contains   a full-dimensional ball of radius $r$ and is contained inside an origin-centered ball of radius $R$. Consequently, its objective is $\lipschitzcost$-Lipschitz for some constant $\lipschitzcost$. 
We now state the solution of \Cref{eq:reformulated}~\cite{alessio2009survey} and later (in \Cref{thm:convex_combination}) show how it appears in the smoothness of the \textit{barrier} MPC solution.
\begin{fact}[\cite{bemporad2002explicit}]\label{thm:dudx}
Let $\sigma \in \{0,1\}^\nconstr$ denote a set of active constraints for \Cref{eq:reformulated}, with $\sigma_i= 1$ iff the $i$th constraint is active. We overload this notation so that $\sigma(x_0)$ represents active constraints of the solution of \Cref{eq:reformulated} for a particular $x_0$. Let $P_\sigma = \{ x | \sigma(x) = \sigma\}$ be the
the set of $x_0$ for which the solution has active constraints $\sigma$. Then for $x_0 \in P_\sigma$, the solution $u$ of \Cref{eq:reformulated} may be expressed as
$ u = K_{\sigma} x_0 + k_{\sigma}$, 
where $K_\sigma$ and $k_\sigma$ are defined as: 
\begin{equation}\numberthis\label{eq:K_sigma}
\begin{aligned}
    K_{\sigma} &:= \mH^{-1}[\mF^\top - \mG^\top(\mG \mH^{-1} \mG^\top)_\sigma^{-1}(\mG \mH^{-1} \mF^\top - \mP)],\\
    k_\sigma &:= H^{-1} G^\top (G H^{-1} G^\top)_\sigma^{-1} w.
\end{aligned}
\end{equation}
\end{fact} 
Based on this fact, one may pre-compute an efficient lookup structure mapping $x \in P_\sigma$ to $K_\sigma, k_\sigma$. However, since every combination of active constraints may potentially yield a unique feedback law, the number of pieces to be computed may grow \textit{exponentially} in the problem dimension or time horizon. For instance, even the simple two-dimensional toy system in \Cref{fig:explicit_mpc} has  $261$ pieces. In high dimensions or over long time horizons, merely enumerating all pieces of the explicit MPC may be computationally intractable. 

This observation motivates us to consider approximating explicit MPC using a polynomial number of sample trajectories, collected offline. We introduce this framework next.

\section{Motivating Smoothness: Imitation Learning Frameworks}\label{sec:learning_guarantees}
In this section, we motivate barrier MPC by specializing  to the setting of \Cref{eq:V} the framework from \cite{pfrommer2022tasil}, which enables high-probability guarantees on the quality of an approximation.  

Suppose we are given an expert controller $\piexpert$, a policy class $\Pi$, a distribution of initial conditions $\mathcal{D}$, and $N$ sample trajectories $\{x_{0:K-1}^{(i)}\}_{i=1}^N$ of length $K$, with $\{x_0^{(i)}\}_{i=1}^N$ sampled i.i.d from $\mathcal{D}$. Our goal is to find an approximate policy $\widehat{\pi} \in \Pi$ such that, given an accuracy parameter $\epsilon$, the closed-loop states $\widehat{x}_t$  and $x_t^\star$ induced by $\widehat{\pi}$ and $\piexpert$, respectively, satisfy, with high probability over $x_0 \sim \mathcal{D}$,
$$\|\widehat{x}_t - x^\star_t\| \leq \epsilon, \forall t > 0.$$ This is formalized in \Cref{prop:goodness_of_learned_policy}. To understand this statement, we first establish some assumptions. 

We first assume through  \Cref{assumption:closeness_approx_expert} that $\pilearned$ has been chosen by a black-box supervised imitation learning algorithm which, given the input data, produces a $\pilearned \in \Pi$ such that, with high probability over the  distribution induced by $\mathcal{D}$, the policy {and its Jacobian} are close to the expert.

\begin{assumption}\label{assumption:closeness_approx_expert}
For some $\delta \in (0,1), \epsilon_0 > 0, \epsilon_1 > 0$ and given $N$ trajectories $\{x_{0:K-1}^{(i)}\}_{i=1}^{(N)}$ of length $K$ sampled i.i.d. from $\mathcal{D}$ and rolled out under $\piexpert$, the approximating policy $\pilearned$ satisfies:
\begin{align*}
    \mathbb{P}_{x_0 \sim \mathcal{D}}\bigg[&\sup_{k \geq 0}\|\pilearned(x_k) - \piexpert(x_k)\| \leq \epsilon_0/N \,\,\, \wedge \,\,\,  \sup_{k \geq 0}\left\|\frac{\partial\widehat{\pi}}{\partial x}(x_k) - \frac{\partial \piexpert}{\partial x}(x_k)\right\| \leq \epsilon_1/N \bigg] \geq 1 - \delta. 
\end{align*}\label{assum:bounds}
\end{assumption}
For instance, as shown in \cite{pfrommer2022tasil}, \Cref{assumption:closeness_approx_expert} holds for $\pilearned$ chosen as an empirical risk minimizer from a class of twice differentiable parametric functions with $\ell_2$-bounded parameters, e.g., dense neural networks with smooth activation functions and trained with $\ell_2$ weight regularization.  We refer the reader to \cite{pfrommer2022tasil, tu2022sample} for other such examples of $\Pi$. 
Note the above definition   requires  generalization on only the state distribution induced by the expert, rather than the distribution induced by the learned policy, as in \cite{ahn2023model, chen2018approximating}.

Next, we define a weaker variant of the standard \emph{incremental input-to-state stability} ($\delta$ISS) \cite{vosswinkel2020determining} and assume, in \Cref{assum:stable}, that this property holds for the expert policy. 

\begin{definition}[Local Incremental Input-to-State Stability, cf. \cite{pfrommer2022tasil}]\label{def:locIncStab} For all initial conditions $x_0 \in X$ and bounded sequences of input perturbations $\{\Delta_t\}_{t > 0}$ that satisfy $\|\Delta_t\| < \eta$, let $\overline{x}_{t+1} = \fcl{\pi}(\overline{x}_t, 0)$, $\overline{x}_0 = x_0$ be the nominal trajectory, and let $x_{t+1} = \fcl{\pi}(x_t, \Delta_t)$ be the perturbed trajectory. We say that the closed-loop dynamics under $\pi$ is $(\eta, \gamma)$-locally-incrementally stable for $\eta, \gamma > 0$ if 
\begin{align*}
    \|x_t - \overline{x}_t\| \leq \gamma \cdot \max_{k < t} \|\Delta_k\|, \quad \forall t \geq 0.
\end{align*}
\end{definition}

\begin{assumption}\label{assum:stable}
    The expert policy $\piexpert$ is ($\eta,\gamma$)-locally incrementally stable.
\end{assumption}
As noted in \cite{pfrommer2022tasil}, local $\delta$ISS is a much weaker criterion than even just regular incremental input-to-state stability. There is considerable prior work demonstrating that ISS (and $\delta$ISS) holds under mild conditions for both the explicit MPC and the barrier-based MPC under consideration in this paper \cite{pouilly2020stability}. We refer the reader to \cite{zamani2011lyapunov} for more details.
Having established some preliminaries for stability, we now move on to the smoothness property we consider. 
\begin{definition}[Smoothness]\label{def:def_smoothness} We say that an MPC policy $\pi$ is $(L_0, L_1)$-smooth if for all $x_0 \in \X$ and $x_1 \in \X$, 
\begin{align*}
    \|\pi(x_0) - \pi(x_1)\| &\leq L_0\|x_0 - x_1\|, \\
    \left\|\partial_x \pi(x_0) - \partial_x \pi(x_1)\right\| &\leq L_1\|x_0 - x_1\|.
\end{align*}
\end{definition}
\begin{assumption}\label{assumption:smoothness_of_exp_and_learned} The expert policy $\piexpert$ and the learned policy $\pilearned$ are both $(L_0, L_1)$-smooth.
\end{assumption}
At a high level, by assuming smoothness of the expert and the learned policy, we can implicitly ensure that the learned policy captures the stability of the expert in a neighborhood around the data distribution. If the expert or learned policy were to be only piecewise smooth, a transition from one piece to another in the expert, which is not replicated by the learned policy, could  lead to  unstable closed-loop behavior.

Having stated all the necessary assumptions, we are now ready to state below  the main export of this section, guaranteeing closeness of the learned and expert policies. 

\begin{fact}[cf. \cite{pfrommer2022tasil}, Corollary A.1]\label{prop:goodness_of_learned_policy} Provided $\piexpert, \pilearned$ are $(L_0,L_1)$-smooth, $\piexpert$ is $\gamma$-locally-incrementally stable, and $\widehat{\pi}$ satisfies \Cref{assum:bounds} with $\frac{\epsilon_0}{N} \leq \frac{1}{16\gamma^2 L_1}$ and $\frac{\epsilon_1}{N} \leq \frac{1}{4\gamma}$, $\delta > 0$, then with probability $1-\delta$ for $x_0 \sim \mathcal{D}$, we have 
\begin{align*}
    \|\widehat{x}_t - x^\star_t\| \leq \frac{8 \gamma \epsilon_0}{N} \quad \forall t \geq 0.
\end{align*}
\end{fact}
The upshot of this result is that to match the trajectory of the MPC policy $\piexpert$ with high probability, provided $\piexpert$ is $(L_0, L_1)$-smooth, {we  need to match the Jacobian and value of $\piexpert$ on \emph{only} $N  K$ pieces.} 
This is in contrast to prior work such as \cite{maddalena2020neural, karg2020efficient, chen2018approximating} on approximating explicit MPC, which require sampling new control inputs during training (in a reinforcement learning-like fashion) or post-training verification of the stability properties of the network.

However, as we noted in \Cref{assumption:smoothness_of_exp_and_learned}, \textit{these strong guarantees crucially require a smooth expert controller}. In the following sections, we investigate two approaches for smoothing $\pimpc$: randomized smoothing and barrier MPC. 

\subsection{Randomized Smoothing}
We first consider randomized smoothing \cite{duchi2012randomized} as a baseline 
approach for smoothing $\piexpert$.  Here, the imitator $\pi^{\mathrm{rs}}$ is learned with a loss function that randomly samples  
noise drawn from a  probability distribution chosen to smooth the policy. This approach corresponds to the following controller.
\begin{definition}[Randomized Smoothed MPC]\label{def:rand_smoo} Given a control policy $\pimpc$ of the form \cref{eq:pi_mpc}, a desired zero-mean noise distribution $\mathcal{P}$, and magnitude $\epsilon > 0$, the randomized-smoothing based MPC is defined as:
\begin{align*}
    \pi^{\mathrm{rs}}(x) := \E_{w \sim \mathcal{P}}[\pimpc(x+ \epsilon w)].
\end{align*}
\end{definition}

The distribution $\mathcal{P}$ in \Cref{def:rand_smoo} is chosen such that the following guarantees on error and smoothness hold. 

\begin{fact}[c.f. \cite{duchi2012randomized}, Appendix E, Lemma 7-9] \label{thm:randomized_bounds}
For $\mathcal{P} \in  \{\mathrm{Unif}(B_{\ell_2}(1)), \mathrm{Unif}(B_{\ell_\infty}(1)), \,\mathcal{N}(0,I)\}$, there exist $L_0, L_1$ that depend on $d_x$ and the Lipschitz constant of $\pimpc$ such that 
\begin{align*}
    \|\pi^{\mathrm{rs}}(x) - \pi^{\mathrm{mpc}}(x)\| &\leq L_0 \epsilon \quad \forall x \in \X, \\
    \|\nabla \pi^{\mathrm{rs}}(x) - \nabla \pi^{\mathrm{rs}}(y)\| &\leq \frac{L_1}{\epsilon}\|x - y\| \quad \forall x,y \in \X.
\end{align*}
\end{fact}
Using randomized smoothing to obtain a smoothed policy has the following key disadvantages: Firstly the expectation $ \E_{w \sim \mathcal{P}}[\pimpc(x+ \epsilon w)]$ is evaluated via sampling, which means the policy must be continuously re-evaluated during training in order to guarantee a smooth learned policy. Secondly, smoothing in this manner may cause $\pi^{\mathrm{rs}}$ to violate state constraints. Finally, simply smoothing the policy may not preserve the stability of $\pimpc$. As we shall show, using barrier MPC as a smoothed policy overcomes all these drawbacks.

\section{Our Approach to Smoothing: Barrier MPC}\label{sec:barrier_mpc_all}
Having described the guarantees obtained via randomized smoothing, we now consider smoothing via {barrier} functions.
\begin{definition}[\cite{nesterov1994interior}]\label{def:sc_and_scb}
Given an {open} convex set $Q\subseteq \mathbb{R}^n$, a function $f:Q\mapsto\mathbb{R}$ is  self-concordant on  $Q$ if for any $\vx\in Q$ and any direction $h\in \mathbb{R}^n$, the following inequality holds: 
\[
\lvert \mathcal{D}^3 f(\vx)[h, h, h]\rvert \leq 2(\mathcal{D}^2 F(\vx)[h, h])^{3/2}, 
\] 
where $\mathcal{D}^k f(\vx)[h_1, \dotsc, h_k]$ is the $k^\mathrm{th}$ derivative of $f$ at $\vx$ along  directions $h_1,\dotsc,h_k$. {Moreover,} $f$ is a $\nu$-self-concordant barrier on $Q$ if { it further satisfies  $\lim_{x\to \partial Q} f(x) = + \infty$} and the inequality  $\nabla f(\vx)^\top (\nabla^2 f(\vx))^{-1} \nabla f(\vx) \leq \nu$ for any $x\in Q$.
\end{definition}

The self-concordance property essentially says that locally, the Hessian does not change too fast --- it has therefore proven extremely useful in interior-point methods to design fast algorithms for (constrained) convex programming~\cite{zbMATH03301975, karmarkar1984new} and has also found use in model-predictive control \cite{wills2004barrier}.

We consider using {barrier MPC} as a natural alternative to randomized smoothing of \Cref{eq:pi_mpc}. In barrier  MPC, the inequality constraints  in the  optimal control
problem are eliminated by incorporating them into the cost
function via suitably scaled barrier terms. In this paper, we work only with the log-barrier, which turns a  constraint $f(x)\geq 0$ into the term $-\eta \log(f(x))$ in the minimization objective and is the standard choice of barrier on polytopes~\cite{nesterov1994interior}. 

Concretely, starting from our MPC reformulation in \Cref{eq:reformulated}, the barrier MPC we work with is defined as follows.

\begin{figure*}
    \centering
    \includegraphics[width=1\linewidth]{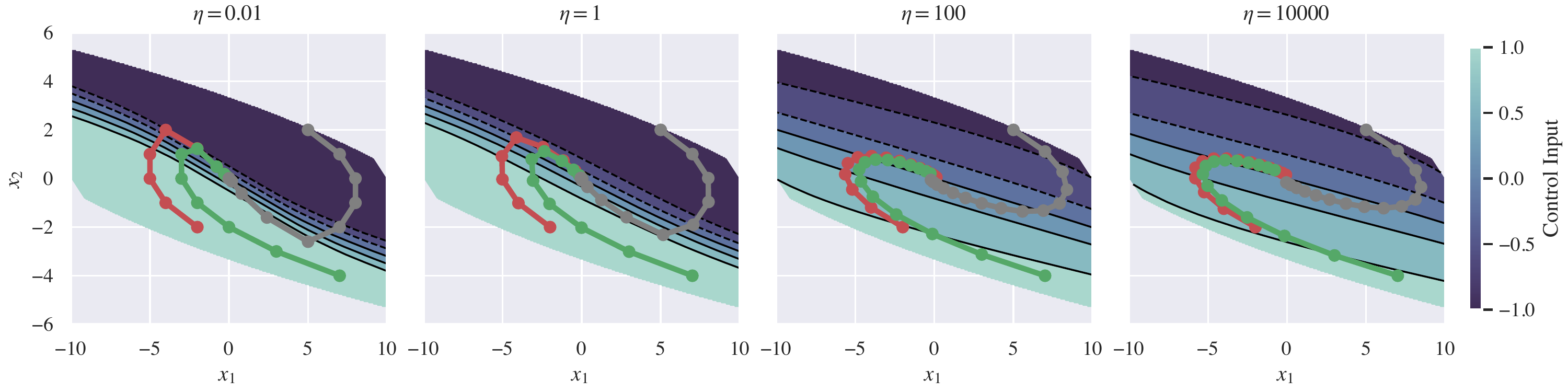}
    \caption{Visualizations of the log-barrier MPC control policy and several trajectories for the same system as \Cref{fig:explicit_mpc} and different choices of $\eta$.}
    \label{fig:smoothing_contours}
\end{figure*}

\begin{problem}[Barrier MPC]\label{def:barr_mpc_formal} Given an MPC as in \Cref{eq:reformulated} and  weight $\eta > 0$, the barrier MPC is defined by minimizing, over the input sequence $\vueta\in\R^{\T\cdot d_u}$, the cost
\[ 
\begin{array}{ll}
\vspace{0.5em}
V^\eta(x_0, \vueta)&:= \tfrac{1}{2}\vueta^\top \mH \vueta  - \vx_0^\top \mF \vueta - \eta \left[1^\top\log(\phi(x_0, \vueta)) - \vd^\top \vueta    \right],
\end{array}\numberthis\label{eq:barrierMPC}\] where $\phi_i(x_0, \vueta)$ (see \Cref{eq:phi}) is the residual of the $i^\mathrm{th}$ constraint for $\vx_0$ and $\vueta$, and choosing $\vd :=\nabla_{\vueta} \sum_{i=1}^m\log(\phi_i(0, \vueta)) \vert_{\vueta=0}$ ensures $\argmin_{\vueta} V^\eta(0, \vueta) = 0.$ We overload $\vueta(x_0)$ to also denote the minimizer of \Cref{eq:barrierMPC} for a given $x_0$ and use $\pibmpc(x) := \argmin_{u_0} \min_{u_{1:\T-1}} V^\eta(x,u)$ for the associated control policy.
\end{problem}

The following result, based on standard techniques to analyze the sub-optimality gap in interior-point methods, bounds the distance between  the optimal solution of \Cref{def:barr_mpc_formal} and that of explicit MPC in \Cref{eq:reformulated}.  

\begin{theorem} 
 Suppose that $\vueta$ and $u^\star$ are, respectively, the optimizers of \Cref{def:barr_mpc_formal} and \Cref{eq:reformulated}. Then we have the following bound in terms of $\eta$ in  \Cref{eq:barrierMPC}:
\[  \|\vueta - \vu^\star\|\leq O(\sqrt{\eta}). \]
\end{theorem}
\begin{proof} In this proof, we use $K$ for the constraint polytope of \Cref{eq:reformulated}. First,  \Cref{lem:linear_plus_barrier_sc}  establishes that the recentered log-barrier in \Cref{def:barr_mpc_formal} is also a self-concordant barrier with some self-concordance parameter $\nu$. 
Since $\vueta= \arg\min_u q(u) + \eta \phi_K(u)$, where $q$ is the quadratic cost function of \Cref{def:barr_mpc_formal} and $\phi_K$ is the centered log barrier on $K$, we have by first-order optimality:  \[ \nabla q(\vueta) = - \eta \nabla\phi_K(\vueta). \numberthis\label{eq:opt_ueta_err_bound}\] Denote by $\alpha$ the strong convexity parameter of the cost function in  \Cref{def:barr_mpc_formal} and by $\nu$ the self-concordance parameter of the barrier $\phi_K$. 
Then,  
\begin{align*}
    \left\{ q(\vueta) - q(u^\star) \right\} + \frac{1}{2}\alpha \|\vueta-u^\star\|_2^2 &\leq \nabla q(\vueta)^\top (\vueta - u^\star)\\ 
    &= \eta \cdot\nabla \phi_K(\vueta)^\top (u^\star - \vueta) \\ 
    &\leq \eta \nu, 
\end{align*} by $\alpha$-strong convexity of $q$,  \Cref{eq:opt_ueta_err_bound}, and applying \Cref{thm:inner_prod_ub_nu}. 
Since both $q(\vueta) - q(u^\star)$ and $\frac{1}{2}\alpha\|\vueta-u^\star\|_2^2$ are positive, we can bound the latter   by $\eta \nu$ to finish the proof. 
\end{proof}

We now proceed to establish the following technical lemma, which we later use in our key smoothness result. 

\begin{lemma}\label{lem:dueta_dx}
The solution to the barrier formulated MPC  in \Cref{def:barr_mpc_formal} evolves with respect to $x_0$ as 
\begin{align*}\label{eq:DuetaDx}
    \DuetaDx &= \mH^{-1}[\mF^\top - \mG^\top(\mG \mH^{-1}\mG^\top + \mLambda)^{-1}(\mG \mH^{-1} \mF^\top - \mP)],
\end{align*} where $\mLambda:= \eta^{-1} \cdot \Phi^{2}$, with $\Phi := \Diag(\phi(x_0, \vueta(x_0)))$ being the diagonal matrix constructed via the residual $\phi(x_0, \vueta(x_0))$ as defined in \Cref{eq:phi}.
\end{lemma}
\begin{proof}
The optimality condition associated with minimizing \Cref{eq:barrierMPC} is:
\[ \mH \vueta(x_0)- \mF^\top \vx_0 + \eta \sum_{i = 1}^m \left(\frac{\vg_i}{\phi_i(x_0, \vueta(x_0))} + \vd_i\right) = 0.\]
Differentiating with respect to $\vx_0$ and rearranging yields \[\DuetaDx = (\mH + \eta \mG^\top \Phi^{-2} \mG)^{-1} (\mF^\top + \eta \mG^\top \Phi^{-2} \mP),\] which upon applying \Cref{fact:shermanMorrisonWoodbury} 
and plugging in $\mLambda$ yields the claimed rate. 
\end{proof} 

We are now ready to state \Cref{thm:convex_combination}, where we connect the rates of evolution of the solution \Cref{eq:K_sigma} to the constrained MPC and that in \Cref{lem:dueta_dx} of the barrier MPC. Put simply, this result tells us that {solving barrier MPC implicitly interpolates between a potentially exponential number of affine pieces from the original explicit MPC problem}. This important connection helps us get a handle on the smoothness of barrier MPC as the rate at which this interpolation changes.

\begin{lemma}\label{thm:convex_combination} With $K_\sigma$ as defined as in \Cref{thm:dudx},   $h_\sigma = \det([G H^{-1}G^\top]_\sigma)\prod_{i=1}^m (\eta^{-1}\phi^2_i)^{1 - \sigma_i}$ from \Cref{lem:split_into_adj}, and the set $S := \{\sigma \, | \, \det([GH^{-1}G^\top]_\sigma) > 0\}$, 
the rates of evolution of the solutions to the constrained MPC (in \Cref{eq:K_sigma}) and the barrier MPC (in \Cref{lem:dueta_dx}) 
are connected as:
\begin{align*}
\DuetaDx = \frac{1}{\sum_{\sigma \in S}h_{\sigma}}\sum_{\sigma \in S} h_\sigma K_\sigma.
\end{align*}
\end{lemma}
\begin{proof}
Applying \Cref{lem:split_into_adj} to $(GH^{-1}G^\top + \Lambda)$ in the expression for $\DuetaDx$ from \Cref{lem:dueta_dx} yields: 
\begin{align*}
    \DuetaDx &= \sum_{\sigma \in S} \frac{h_\sigma}{h} H^{-1}[\mF^\top - \mG^\top(\mG \mH^{-1}\mG^\top)_\sigma^{-1}(\mG \mH^{-1} \mF^\top - \mP)] \\
    &-\sum_{\sigma \in \{0,1\}^m\setminus S} \frac{c_\sigma}{h} \mH^{-1}\mG^\top \adj(\mG \mH^{-1} \mG^\top)_\sigma(\mG\mH^{-1} \mF^\top - \mP)\\
     &= \sum_{\sigma \in S} \frac{h_\sigma}{h} H^{-1}[\mF^\top - \mG^\top(\mG \mH^{-1}\mG^\top)_\sigma^{-1}(\mG \mH^{-1} \mF^\top - \mP)]. 
\end{align*}  
where, as defined in \Cref{lem:split_into_adj},  $h = \sum_{\sigma \in S} h_\sigma$, and $c_\sigma = \prod_{i=1}^m (\eta^{-1}\phi_i^2)^{1-\sigma_i}$.  
The second equality follows since for $\sigma \in \{0,1\}^m \setminus S$ and $\mH \succ 0$, by definition, $\det([GH^{-1}G^\top]_\sigma) = 0$, implying $G^\top \adj(G H^{-1}G^\top)_\sigma = 0$. Finally, plugging in  $K_\sigma$ from \Cref{eq:K_sigma}, one may finish the proof.  
\end{proof}
The above result immediately implies the following upper bound on  $\left\|\DuetaDx\right\|$, independent of $\eta$. 
\begin{corollary}\label{cor:first_der_bound_ueta}
In the setting of \Cref{thm:convex_combination}, we have 
\begin{align*}
    \left\|\DuetaDx\right\| \leq L := \max_{\sigma \in S} \|K_\sigma\|. 
\end{align*}
\end{corollary}
\begin{proof}
    From \Cref{thm:convex_combination}, we can conclude that $\frac{\partial \vueta}{\partial x_0}$  lies in the convex hull of $\{K_\sigma\}_{\sigma \in S}$, and note that $|S| < \infty$.
\end{proof}

The main export of this section, which shows that $\vueta$ (and hence $\pibmpc$) satisfies the conditions of \Cref{assumption:smoothness_of_exp_and_learned}, is the following theorem. This result hinges on  \Cref{thm:res_lower_bound}, which quantifies  lower bounds on residuals when minimizing a convex cost over a polytope, a result we hope could be of independent interest to the optimization community. 
\begin{theorem}\label{thm:hess_ueta_bounded}The Hessian of the solution $\vueta$ of \Cref{def:barr_mpc_formal} with respect to $x_0$ is bounded by:
\begin{align*}
    \left\|\frac{\partial^2 u^\eta}{\partial x_0^2}\right\| \leq \frac{C}{\philb}(\|P\| + \|G\|L)^2,
\end{align*} 
where $C := \max_{\sigma \in S} \|2H^{-1}G^\top (G H^{-1} G^\top)_\sigma^{-1}\|$ (with $S$ as in \Cref{thm:convex_combination}),  matrices $P$ and $G$ are as defined in \Cref{eq:reformulated}, $L$ as in \Cref{cor:first_der_bound_ueta}, and $\philb \geq \min\left\{ \frac{\eta}{2}, \frac{r\eta^2}{150 (\nu\eta^2 + R^2  (\lipschitzcost^2 +1) )} \right\}$, 
where $r$, $R$, and $L_V$ are the inner radius, outer radius, and Lipschitz constant of \Cref{eq:reformulated} as described in \Cref{sec:explicitMPC}, and $\nu=20(\nconstr+R^2\|d\|^2)$, where $\|d\|^2$ in \Cref{def:barr_mpc_formal} is a constant by construction. 
\end{theorem}
\begin{proof}
Let $y \in \R^{d_x}$ be an arbitrary unit-norm vector, and define the univariate function $M(t) := GH^{-1}G^\top + \eta^{-1}\Phi(x_0 + ty, u_\eta(x_0 + ty))^2$ where $u_\eta$ is the solution to \Cref{def:barr_mpc_formal} and  $\Phi := \Diag(\phi)$, with $\phi$ as  in \Cref{eq:phi}. Then by differentiating $M(t)^{-1}$ and applying the chain rule, we get  
\begin{align*}
    \frac{d}{dt}\left(\frac{\partial u^\eta}{\partial x_0}(x_0 + t y)\right) 
    = &\,H^{-1}G^\top M(t)^{-1}\frac{dM(t)}{dt} M(t)^{-1}(GH^{-1} F^\top - P) \\
    = &\, 2H^{-1}G^\top M(t)^{-1}  \frac{d\Phi}{dt} (\eta G H^{-1}G^\top\Phi^{-1} + \Phi)^{-1}(GH^{-1} F^\top - P).
\end{align*}
Applying \Cref{lem:split_into_adj} to $M(t)^{-1}$ implies, $$\|2H^{-1}G^\top M(t)^{-1}\| < C := \max_{\sigma \in S} \|2H^{-1}G^\top (G H^{-1} G^\top)_\sigma^\dagger\|.$$
To bound the other terms in the product, we first note that by arguments about the norm, $\|(\eta GH^{-1} G^\top \Phi^{-1} + \Phi)^{-1}\| \leq \frac{1}{\min_{i \in [\nconstr]} \phi_i}$. Next, the definition of $\Phi$, triangle inequality, and \Cref{thm:convex_combination} give 
$\left\|\frac{d\Phi}{d t}\right\| \leq \|P\| + \|G\| \left\|\frac{\partial u_\eta}{\partial x_0}\right\| \leq \|P\| + \|G\| L$. Finally, recognizing $H^{-1} F^\top$ as $K_{\sigma}$ from \Cref{eq:K_sigma} (with $\sigma = 0 \in \R^{\nconstr}$) yields $ \|G H^{-1}F^\top - P\| = \left\|G K_{0} - P\right\| \leq \|P\| + \|G\|L$. Finally, we combine these   with the lower bound on $\philb$ from \Cref{thm:res_lower_bound} that uses $\nu = 20(m+R^2\|d\|^2)$, the self-concordance parameter (computed via \Cref{lem:linear_plus_barrier_sc}) of the recentered log-barrier in \Cref{def:barr_mpc_formal}. 
\end{proof}
Thus, \Cref{thm:hess_ueta_bounded} establishes bounds analogous to those  in \Cref{thm:randomized_bounds} for randomized smoothing, demonstrating that the Jacobian of the smoothed expert policy is sufficiently Lipschitz. Indeed, in this case our result is stronger, showing that the Jacobian is differentiable and  the Hessian tensor is bounded. This theoretically validates the  core proposition of our paper:  the barrier MPC policy in \Cref{def:barr_mpc_formal} is suitably smooth, and therefore the  guarantees in  \Cref{sec:learning_guarantees} hold.
Having established our theoretical guarantees, we now turn to demonstrating their efficacy in our experiments.

\begin{figure*}
    \centering
    \includegraphics[width=0.95\linewidth]{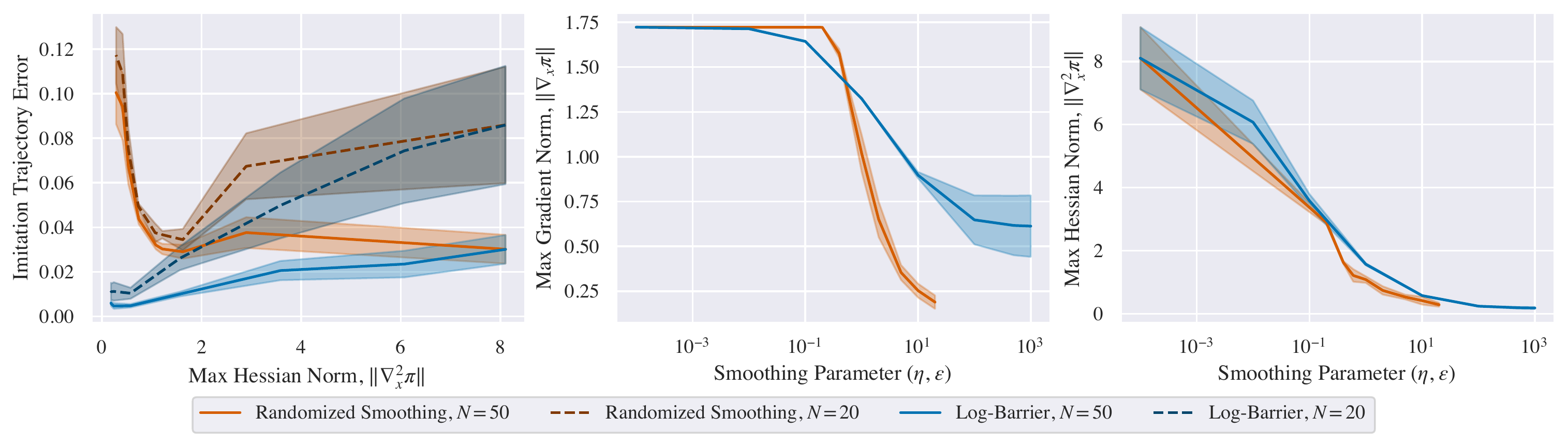}
    \caption{Left: The imitation error $\max_{t} \|\hat{x} - x^\star\|$ for the trained MLP over 5 seeds, as a function of the expert smoothness for both randomized smoothing and log-barrier MPC. Center, Right: The $L_0$ (gradient norm) and $L_1$ (Hessian norm) smoothness of $\pi^\star$ as a function of the smoothing parameter.}
    \label{fig:experiments}
\end{figure*}
\section{Experiments}
We demonstrate the advantage of barrier MPC over randomized smoothing for the toy double integrator system visualized in \Cref{fig:explicit_mpc}. We sample $N \in [20, 50]$ trajectories of length $K = 20$ using $\pibmpc$ and $\pirs$ with a horizon length $T = 10$ and smoothing parameters $\eta$ and $\epsilon$ ranging from $10^{-4}$ to $10^3$ and $10^{-4}$ to $20$, respectively. We use $\mathcal{P} = \mathcal{N}(0,I)$ for the randomized smoothing distribution. For each parameter set, we trained a 4-layer multi-layer perceptron (MLP) using GELU activations \cite{hendrycks2016gaussian} to ensure smoothness of $\Pi$.

In \Cref{fig:experiments}, we visualize the smoothness properties of the chosen $\pi^*$ for each method across the choices of $\eta, \epsilon$. We also show the imitation error $\max_t \|x^\star - \hat{x}\|$ for the learned MLP as a function of the expert smoothness. For more smooth experts, we observe that barrier MPC outperforms randomized smoothing, even in the lower-data setting. These experiments confirm our hypothesis:  barrier MPC is an effective smoothing technique (preserving both stability and constraints) that outperforms randomized smoothing.

\section*{Acknowledgements}
We gratefully acknowledge funding from  ONR N00014-23-1-2299 and a Vannevar Bush Fellowship from the Office of Undersecretary of Defense for Research and Engineering (OUSDR\&E).

\clearpage

\printbibliography

\newpage
\begin{appendix}
\section{Notation and Preliminaries}\label{sec:app-matrix-formulas} 
We use uppercase letters for matrices and lowercase letters for vectors. 
We use $\ve_i$ to denote the vector with one at the $i^\mathrm{th}$ coordinate and zeroes at the remaining coordinates. 
We collect the following relevant facts from matrix analysis. 

\begin{fact}[\cite{horn2012matrix}]\label{fact:block-matrix-determinant}
Consider a matrix $\mA = \begin{bmatrix} 
\va & \vb^\top \\ \vb &  \mD\end{bmatrix}.$ Then $\adj(\mA)$ is defined to be the matrix that satisfies $\adj(\mA) \cdot \mA = \mA\cdot\adj(\mA) = \det(\mA) \cdot \mI$ and  equals the transpose of the cofactor matrix. 
The matrix determinant lemma lets us  express, for any $\mM$, the determinant for a unit-rank update: \[ \det(\mM + \vu\vv^\top) = \det(\mM) + \vv^\top\adj(\mM) \vu\numberthis\label{eq:matDetLem}.\] Notably, \Cref{eq:matDetLem} does \emph{not} require invertibility of $\mD$.
Applying \Cref{eq:matDetLem} to $\mA$ defined above gives: 
\begin{align*}\det(\mA) &= \det\left(\begin{bmatrix}\va & \vb^\top \\ 0 & \mD\end{bmatrix}  + \begin{bmatrix}0 \\ b\end{bmatrix}\ve_1^\top\right) \\
&= \det\left(\begin{bmatrix}a & b^\top \\ 0 & D\end{bmatrix}\right) + \begin{bmatrix}0 & \vb^\top\end{bmatrix} \adj(A) \ve_1\\
&= a \cdot\det(D) - \vb^\top \adj(D) \vb,\numberthis\label{eq:detSchur}
\end{align*} where  the final step is by \Cref{lem:adj_A}.
\end{fact}

\begin{fact}[Woodbury matrix identity]\label{fact:shermanMorrisonWoodbury}
    Given conformable matrices $\mA, \mC, \mU, $ and $\mV$ such that $\mA$ and $\mC$ are invertible, 
\[ (\mA + \mU \mC \mV)^{-1} = \mA^{-1} - \mA^{-1} \mU (\mC^{-1} + \mV \mA^{-1} \mU ) \mV \mA^{-1}.\]
\end{fact}

\subsection{Our Matrix Lemmas}

We first state some technical results that we build upon to prove our first key result (\Cref{lem:split_into_adj}).

\begin{lemma}\label{lem:adj_A}
Consider matrix $\mA$ as in \Cref{fact:block-matrix-determinant}. Then
\[\adj (\mA) = \begin{bmatrix}
    \det(\mD) & -  \vb^\top \adj(\mD) \\
    -\adj(\mD) \vb & a\cdot \adj(\mD) + \mK
\end{bmatrix},\]
where $\mK$ satisfies $DK = \adj(\mD)\vb\vb^\top  - \vb^\top \adj(\mD)\vb \cdot \mI$.
\end{lemma}
\begin{proof}We use that $\adj(A) = \mC^\top$, where $\mC$ is the matrix of cofactors of $\mA$. Let $M_{i,j}$ be the $\{i,j\}^{\mathrm{th}}$ minor of $D$, $M_{j}$ be $D$ with the $j$th column removed, and $v_i$ be $b$ with the $i^\mathrm{th}$ index removed. Then, computing the relevant cofactors gives
$C_{1,1} = \det(D)$,  $C_{1, \space 1+j} = (-1)^{j}\det\left(\begin{bmatrix}
    b & M_j
    \end{bmatrix}\right)= -[b^\top \adj(D)]_j$, and 
\newcommand{\remindex}[2]{#1_{\backslash \{#2\}}}
\begin{align*}
    C_{1 + i, \space 1 + j} &= (-1)^{i + j}\det\left(\begin{bmatrix}
        a & v_j^\top \\
        v_i & M_{ij}
    \end{bmatrix}\right) \\
    &= (-1)^{i + j}a \cdot \det(M_{ij}) - (-1)^{i + j}v_i^\top \adj(M_{ij}) v_j \\
    &= a\cdot \adj(D)_{ij} - (-1)^{i + j}v_i^\top \adj(M_{ij}) v_j.
\end{align*} By mapping these cofactors back into the definition of the adjugate we want, one can  then conclude the proof, where $K$ collects the $- (-1)^{i + j}v_i^\top \adj(M_{ij})v_j$ terms. 
\end{proof}

\begin{fact}[Theorem $2.3$ of \cite{ipsen2008perturbation}]\label{lem:det_A_Lambda}
Given $\mA \in \R^{\nconstr \times \nconstr}$
as in \Cref{fact:block-matrix-determinant}, positive diagonal matrix $\Lambda=\Diag(\lambda) \in \R^{\nconstr \times \nconstr}$, and $A_\sigma$  denoting the principal submatrix formed by selecting $A$'s rows and columns indexed by $\sigma\in \{0, 1\}^{\nconstr}$, we have
\[\det(\mA + \mLambda) = \sum_{\sigma \in \{0,1\}^n}\left(\prod_{i=1}^m \lambda_i^{1 - \sigma_i}\right)\det(\mA_\sigma).\] 
\end{fact} 

\begin{proof} 
By applying \Cref{eq:matDetLem} and \Cref{lem:adj_A}, we have
\begin{align*} 
\det\left(\mA + \sum_{i=2}^m \lambda_i \ve_i \ve_i^\top + \lambda_1 \cdot \ve_1 \ve_1^\top\right)  
= \; &\det\left(\mA + \sum_{i=2}^m \lambda_i \ve_i \ve_i^\top\right) + \lambda_1 \cdot \det\left(\mD + \sum_{i=1}^{m-1} \lambda_{i+1} \ve_i \ve_i^\top\right).
\end{align*}
The lemma follows by recursive application of \Cref{eq:detSchur} with respect to $\lambda_1, \ldots, \lambda_m$ and noting that the determinant is invariant to permuting both rows and columns.
\end{proof} 

\begin{lemma}\label{lem:split_into_adj} For a positive semi-definite matrix $\mA \in \R^{\nconstr \times \nconstr}$ and a diagonal positive matrix $\mLambda = \Diag(\lambda)$, we have 
\begin{align*}
    (\mA + \mLambda)^{-1} &=  \sum_{\substack{\sigma \in \{0,1\}^m,\\ \det(\mA_\sigma) \neq 0}} \frac{h_\sigma}{h}(\mA_\sigma)^{-1}
    + \sum_{\substack{\sigma \in \{0,1\}^m,\\ \det(\mA_\sigma) = 0}} \left(\frac{\prod_{i=1}^m\lambda_i^{1 - \sigma_i}}{h}\right)\adj(\mA_\sigma) ,
\end{align*}
where $h_\sigma = \det(\mA_\sigma)\prod_{i=1}^m \lambda_i^{1-\sigma_i}$ and $h = \sum_{\sigma \in \{0,1\}^m} h_\sigma$.
\end{lemma}
\begin{proof}From \Cref{lem:adj_A} by observation:
\begin{align*}
\adj \left(\mA + \sum_{i=1}^m \lambda_i \ve_i \ve_i^\top \right) 
= \; &\adj(\mA + \sum_{i=2}^m\lambda_i \ve_i\ve_i^\top) + \lambda_1\begin{bmatrix}0 & 0 \\ 0 & \adj\left(D + \sum_{i=1}^{m-1}\lambda_{i+1} \ve_i\ve_i^\top\right)\end{bmatrix}.
\end{align*}
By repeated application of this fact we have
\[\adj (\mA + \mLambda) = \sum_{\sigma \in \{0,1\}^n} \left(\prod_{i=1}^m \lambda_i^{1-\sigma_i} \right)\adj(\mA_\sigma).\]
The result then follows by application of \Cref{lem:det_A_Lambda} to note that $h = \det(A + \Lambda)$ and casing by invertibility of $\mA_\sigma$.
\end{proof}

\subsection{Results from Convex Analysis}
\begin{fact}[\cite{nesterov1994interior}]\label{thm:inner_prod_ub_nu}
    Let $\Phi$ be a $\nu$-self-concordant barrier. Then for any $x\in \textrm{dom}(\Phi)$ and $y\in \textrm{cl(dom)}(\Phi)$,  \[\nabla \Phi(x)^\top (y-x) \leq \nu.\]  
\end{fact}

\begin{theorem}\label{thm:res_lower_bound}
Let $K=\left\{ \vx:{{A}}\vx\geq{b}\right\} $ be
a polytope such that each of $m$ rows of ${A}$ is normalized to be
unit norm. Let $K$ contain a ball of radius $r$ and be contained
inside a ball of radius $R$ centered at the origin. Let 
\begin{equation}
\ue:=\arg\min_{\u}q(\u)+\eta\phi_{K}(\u),\label{eq:def-ueta}
\end{equation}
where $q$ is a convex $L$-Lipschitz function and $\phi_K$ is a $\nu$-self-concordant barrier on $K$. We show for $\re_i(u_\eta)$, the $i^\mathrm{th}$ residual at $u_\eta$, the following lower bound:
 \[\re_i(u_\eta)\geq  \min\left\{ \frac{\eta}{2}, \frac{r\eta^2}{150 (\nu\eta^2 + R^2 (L^2 +1) )} \right\}.\] 
\end{theorem}

\looseness=-1To prove \Cref{thm:res_lower_bound}, we need the following technical result by \cite{zong2023short}, bounding the optimality gap of a convex program with a linear cost and  a barrier  enforcing its constraints. 

\begin{fact}[\cite{zong2023short}]\label{lem:zong_opt_gap}
    Fix a vector $\vc$, a polytope $K$, and a point $\v.$ We assume
that the polytope $K$ contains a ball of radius $r.$ Let $\vstar=\arg\min_{\u\in K}\vc^{\top}\u$.
We define, for $\vc,$  
\begin{equation}
\gap(\v)=\vc^{\top}(\v-\vstar).\label{eq:1}
\end{equation}
Further, define $\v_{\eta}=\arg\min_{\v}\vc^{\top}\v+\eta\phi_{K}(\v)$,
where $\phi_{K}$ is a self-concordant barrier on $K.$ Then we have the following
lower bound on this suboptimality gap evaluated at $\veta$: 
\begin{equation}
\min\left\{ \frac{\eta}{2},\frac{r\|\vc\|}{2\nu+4\sqrt{\nu}}\right\} \leq\gap(\veta)=\vc^{\top}(\veta-\vstar).\label[ineq]{eq:2}
\end{equation}
\end{fact}

We also need the following technical results from \cite{ghadiri2024improving}. 

\begin{fact}[\cite{ghadiri2024improving}]
\label{lem:f_sc_one_by_nine_r_squared}If $f$ is a self-concordant barrier for a set $K\subseteq B(0,R),$
then $\nabla^{2}f(x)\succeq\frac{1}{9R^{2}}I$ for any $x\in K$. 
\end{fact}

\begin{fact}[\cite{ghadiri2024improving}]
\label{lem:linear_plus_barrier_sc}If $f$ is a $\nu$-self-concordant barrier for a given
convex set $K$ then $g(x)=c^{\top}x+f(x)$ is a self-concordant barrier over $K$. Further, if $K\subseteq B(0,R),$ then $g$ has self-concordance parameter at most 
$20(\nu+R^{2}\|c\|^{2})$. 
\end{fact}

We now prove \Cref{thm:res_lower_bound}. 
\begin{proof}[Proof of \Cref{thm:res_lower_bound}]
Applying the first-order optimality condition of $\ue$ in \Cref{eq:def-ueta} gives us that 
\begin{equation}
\eta\nabla\phi_{K}(\ue)+\nabla q(\ue)=0.\label{eq:a}
\end{equation}
From here on, we fix $\vc=\nabla q(\ue),$where $\ue$ is as in
\Cref{eq:def-ueta}. Then, 
 we may conclude 
\begin{equation}
\ue\in\arg\min_{\u}\vc^{\top}\u+\eta\phi_{K}(\u),\label{eq:4}
\end{equation}
where we have
replaced the cost $q$ in   \Cref{eq:def-ueta} with a specific linear cost $\vc$; to see \Cref{eq:4},
observe that $\ue$ satisfies the first-order optimality condition
of  \Cref{eq:4} because of  \Cref{eq:a} and our choice of $\vc$. 

We now define the function $\phitilde_K(x)= \eta^{-1}\cdot(c - a_i)^\top x + \phi_K(x)$. By \Cref{lem:linear_plus_barrier_sc}, we have that $\phitilde_K$ is a self-concordant-barrier on $K$ with self-concordance parameter 
\[ \widetilde{\nu}\leq 20 (\nu+ R^2\eta^{-2}(\|c\|^2+\|a_i\|^2).\numberthis\label{eq:new-scb-parameter} \] With this new self-concordant barrier in hand, we may now express $\vueta$ from \Cref{eq:4} as the following optimizer: \[ \vueta = \arg\min_u a_i^\top u + \eta \phitilde_K(u).\numberthis\label{eq:new-expression-for-ueta}\] 

Further,
let $\us\in\arg\min_{\u\in K}a_i^\top\u$. 
By applying 
\Cref{lem:zong_opt_gap} to $\vueta$ expressed as in  \Cref{eq:new-expression-for-ueta}, we have 
\begin{equation}
\min\left\{ \frac{\eta}{2},\frac{r\|a_i\|}{2\widetilde{\nu}+4\sqrt{\widetilde{\nu}}}\right\} \leq a_i^{\top}(\ue-\us).\label[ineq]{eq:dist_bound_zeroth}
\end{equation} 

The lower bound in \Cref{eq:dist_bound_zeroth} may be expanded upon via \Cref{eq:new-scb-parameter}, and chaining this with the observation  $a_i^\top (u_\eta- \us) = \textrm{res}_i (\vueta) - \textrm{res}_i(\us)$ gives: \[ \min\left\{ \frac{\eta}{2}, \frac{r\|a_i\|}{150 (\nu + R^2 \eta^{-2} (\|c\|^2 + \|a_i\|^2) )} \right\} \leq \textrm{res}_i(\vueta)-\textrm{res}_i(\us).\] The definition of $\us$ 
implies $\textrm{res}_i(\us) \geq 0$, hence  $\textrm{res}_i(u_\eta)\geq  \min\left\{ \frac{\eta}{2}, \frac{r}{150 (\nu + R^2 \eta^{-2} (L^2 + 1) )} \right\}.$ Repeating this computation for each constraint of $K$ gives the claimed bound overall. 
\end{proof}

\end{appendix}

\end{document}